\documentclass[letterpaper, 10 pt, journal, twoside]{ieeetran}
\IEEEoverridecommandlockouts

\usepackage{amsmath,amssymb,amsfonts}
\usepackage{algorithmic}
\usepackage{caption}
\usepackage{booktabs}
\usepackage{multirow}
\usepackage{hhline}
\usepackage{cite}
\usepackage[utf8]{inputenc}
\usepackage[T1]{fontenc}
\usepackage{graphicx}
\usepackage{textcomp}
\usepackage{comment}
\usepackage{amsthm}
\usepackage{xcolor}
\usepackage{color}
\usepackage{subcaption}
\usepackage{dsfont}
\usepackage{soul}
\usepackage{algorithm}
\usepackage{algorithmic}
\newtheorem{theorem}{Theorem}
\newtheorem{problem}{Problem}

\newtheorem{assumption}{Assumption}

\newtheorem{lemma}{Lemma}
\newtheorem{remark}{Remark}
\newtheorem{definition}{Definition}

\newcommand{\mR}{{\mathbb R}}

\newcommand{\mP}{{\mathbb P}}

\newcommand{\bF}{{\mathbf F}}
\newcommand{\bk}{{\mathbf k}}

\newcommand{\bP}{{\mathbf P}}
\newcommand{\bR}{{\mathbf R}}

\newcommand{\mX}{{\mathbb X}}
\newcommand{\mY}{{\mathbb Y}}

\newcommand{\bM}{{\mathbf M}}

\newcommand{\bX}{{\mathbf X}}

\newcommand{\bx}{{\mathbf x}}
\newcommand{\by}{{\mathbf y}}
\newcommand{\bg}{{\mathbf g}}

\newcommand{\bu}{{\mathbf u}}
\newcommand{\bU}{{\mathbf U}}

\newcommand{\bPsi}{{\mathbf \Psi}}
\newcommand{\bs}{{\mathbf s}}
\newcommand{\bff}{{\mathbf f}}
\newcommand{\bc}{{\mathbf c}}

\newcommand{\brho}{{\boldsymbol{\rho}}}

\usepackage{float}
\usepackage{amsmath}

\begin{document}

\title{Data-Driven Convex Approach to Off-road Navigation via Linear Transfer Operators}
\author{Joseph Moyalan, Yongxin Chen and Umesh Vaidya
\thanks{Financial support from NSF under grants 1942523, 2008513, 2031573 and NSF CPS award 1932458 is greatly acknowledged. J. Moyalan and U. Vaidya are with the Department of Mechanical Engineering, Clemson University, Clemson, SC; {\{jmoyala,uvaidya\}@clemson.edu}. Y.\ Chen is with the School of Aerospace Engineering,
Georgia Institute of Technology, Atlanta, GA; {\{yongchen\}@gatech.edu}}
}

\maketitle

\begin{abstract}
We consider the problem of optimal control design for navigation on off-road terrain. We use a traversability measure to characterize the difficulty of navigation on off-road terrain. The traversability measure captures terrain properties essential for navigation, such as elevation maps, roughness, slope, and texture. The terrain with the presence or absence of obstacles becomes a particular case of the proposed traversability measure. We provide a convex formulation to the off-road navigation problem by lifting the problem to the density space using the linear Perron-Frobenius (P-F) operator. The convex formulation leads to an infinite-dimensional optimal navigation problem for control synthesis. We construct the finite-dimensional approximation of the optimization problem using data. We use a computational framework based on the data-driven approximation of the Koopman operator. This makes the proposed approach data-driven and applicable to cases where an explicit system model is unavailable. Finally, we apply the proposed navigation framework with single integrator dynamics and Dubin's car model.
\end{abstract}

\section{Introduction}
\IEEEPARstart{N}{avigation} problem is one of the most critical research fields in the robotics community. More recently, the problem of off-road navigation, driven by robotics applications in an unstructured environment, has received much attention. The objective is to drive a robot/vehicle from some initial set to the desired target set through a terrain where traversability varies continuously over the entire domain of interest. This is in contrast to navigation in the presence of obstacles where the regions with obstacles are prohibited and hence not traversable. There is extensive literature on navigation in the presence of obstacles. Navigation function and potential function are used for navigation in the presence of obstacles \cite{hirsch2012differential,koditschek1990robot,khatib1985real,latombe2012robot}. While the potential function could have local minima preventing the navigation from initial set to the target, the navigation function is hard to find. The control barrier functions (CBFs) are also used for navigation with safety constraints \cite{ames2019control}. CBFs combine ideas from the control Lyapunov function and barrier certificates for invariance to ensure safety. However, finding CBFs suffer from the same challenges as finding control Lyapunov function and cannot be easily adapted for navigation in off-road terrain where the definition of safety itself is nebulous. 


The problem of off-road navigation has attracted more interest recently. In \cite{langer1994behavior}, perception is used to determine the terrain traversability and local control strategy for navigation. Most of the current literature on this topic has been using sensor data from LIDAR, cameras, and GPS/IMUs to map the off-road terrain to generate a traversability map \cite{overbye2020fast,gao2019off,yang2020multi,hussein2016autonomous}. An existing algorithm such as A$^\star$ is used to design traversable paths in the off-road environment. However, due to the nonconvex nature of the traversability map and hence the cost, the problem becomes nonconvex and, therefore, difficult to solve with no guarantee of global optimality.


One of this paper's main contributions is providing a convex formulation to the off-road navigation problem. The convex formulation is made possible by transforming the problem in the dual space of densities. The formulation of optimal control problem in the dual space of densities is proposed in \cite{huang2022convex,moyalan2023data}, and its extension to navigation problem in the presence of deterministic and stochastic obstacles is studied in \cite{vaidya2018optimal,yu2022data,moyalan2022navigation}. This paper focuses on the off-road navigation problem for a given traversability map. The terrain traversability map includes information about the difficulty level in navigating.
The terrain's traversability measure depends on terrain parameters such as elevation map, roughness, slope, and texture. Therefore, we have utilized the normalized elevation map while constructing the traversability measure. 

The convex formulation leads to an infinite-dimensional convex optimization problem for the off-road navigation problem. First, we use data to construct the finite-dimensional approximation of the infinite-dimensional convex problem. Then, we use a computational framework based on the data-driven approximation of linear Koopman and Perron-Frobenius (P-F) operators for the finite-dimensional approximation of the infinite-dimensional convex optimization problem. The second main contribution is providing a numerically efficient computational algorithm for the data-driven approximation of the P-F operator, preserving some natural properties of this operator.  
Finally, we demonstrate the application of the developed framework for off-road navigation of vehicle dynamics with the Dubin car model. We also compare the results obtained using our proposed approach with the existing A$^\star$ algorithm. The study's main finding is that the traversability cost associated with A$^\star$ is more than one computed using our proposed approach. 

The rest of the paper is structured as follows. Section II  consists of problem formulation, and we discuss the main results in Section III. Then, in Section  IV,  we develop the computational framework based on the linear operator framework. Finally, we present the simulation results in Section V, and a conclusion is in Section VI.
\section{Problem Formulation for Off-road Navigation }
This section defines the traversability map, which will be later used in the convex formulation of the navigation problem. We will also motivate the choice of the cost function for off-road navigation. Let us consider the following dynamical system in control affine form as 
\begin{align}
    \dot \bx={\bf f}(\bx)+\bg(\bx)\bu\label{eq:sys_control}
\end{align}
where $\bx\in \bX\subset \mR^n$ and $\bu\in \bU\subset \mR^m$ are the states and control input respectively. We assume that $\bff(\bx), \bg(\bx) \in \mathcal{C}^1(\bX,\mR^n)$, i.e. the space of continuously differentiable functions on $\bX$. The dynamical control system is assumed to model the control dynamics of the vehicle. The control-affine form is not restrictive, and this will typically be the case for the robotics and vehicle dynamics application \cite{xiao2021learning,gregory2021improving}.

\noindent {\bf Notations:} We consider $\mathcal{B}(\bx)$ to be the Borel $\sigma$-algebra on $\bX$ and ${\cal M}(\bX)$ as the vector space of real-valued measures on ${\cal B}(\bX)$. Let $\mathcal{L}_{\infty}(\bX)$ and $\mathcal{L}_1(\bX)$ be the space of essentially bounded and integrable functions on $\bX$ respectively. The notations in bold and lower case will represent vectors and notations in bold and upper case will represent matrices. Also, $\bs_t(\bx)$ is the notation for the trajectory of feedback system $\dot \bx=\bff(\bx)+\bg(\bx)\bk(\bx)$ starting from initial condition $\bx$ at time $t\in \mR$, where $\bu=\bk(\bx)\in \mathcal{C}^1(\bX,\mR^m)$ is the feedback input. Similarly, $\bs_{-t}(\bx)$ represents the closed-loop trajectory as the function of initial condition $\bx$ backward in time.
\subsection{Traversability Map}\label{trav_map}
We assume that the traversability description of the terrain is captured by a nonnegative function $b(\bx)\in {\cal L}_1(\bX)$. We assume that the function $b(\bx)$ captures the information of the elevation map, terrain roughness, slope, and terrain texture. The construction of such a map is an active area of interest where onboard sensors on the vehicles such as vision, LIDAR, and IMU, as well as drone sensory images, can be used to construct such maps \cite{huang2021decentralized,armbrust2011ravon,bagheri2017development}. We propose the following definition of traversability measure, which captures the relative degree of difficulty of traversing unstructured terrain.
\begin{definition}
Let $\mu_b \in \mathcal{M}(\bX)$ be the associated traversability measure, i.e.,  $d \mu_b(\bx) = b(\bx)d\bx$, where $b(\bx)\geq 0$ is assumed to be an integrable function and is zero on the final target set, $\bX_T$. For any set $A \in \mathcal{B}(\bX)$, the traversability  of the set $A$ is defined using $b(\bx)$ as
\begin{align}
    {\rm Trav}(A)&:=\int_A b(\bx)d\bx=:\mu_b(A).\label{bx_define}
\end{align}

\end{definition}
${\rm Trav} (A)$ captures the relative difficulty of traversing through the region $A\subset \bX$. In particular, if  $\mu_b(A_1)<\mu_b(A_2)$
where $A_i\in {\cal B}(\bX)$, then the region $A_2$ is more difficult to traverse than region $A_1$.  It is easy to see that the above definition of traversability measure also captures the information of binary obstacles. 
 In particular, if $\bX_u$ is an obstacle set, then we can describe it using
\begin{align}
    b(\bx)=\frac{1}{\lambda(\bX_u)}\mathds{1}_{\bX_u}(\bx)
\end{align}
where $\lambda(\cdot)$ is the Lebesgue measure and $\mathds{1}_{\bX_u}$ is the indicator function of the set $\bX_u$. The main objective of this paper involves determining the control inputs $\bu$ to navigate the vehicle dynamics from some initial state $\bX_0$ to some final target set $\bX_T$ while keeping the traversability cost below some threshold, say $ \gamma$, i.e., 
\begin{align}\label{trav_constraint}
    \int_0^\infty b(\bx(t)) dt\leq \gamma
\end{align}
where $\bx(t)$ is the trajectory of the control system (\ref{eq:sys_control}). In this paper, we are interested in the asymptotic navigation problem, where the objective is to find the shortest distance path to the target and the control cost. In particular, we consider the following cost function
\begin{align}\label{cost_1}
  \min_{\bu}\;\; V(\bx)=  \min_{\bu}\;\int_0^\infty q(\bx(t))+\bu^\top \bR\bu dt.
\end{align}
where $q(\bx)$ is the distance function which is zero at the target set $\bX_T$, and $\bR>0$ is the positive definite matrix. Instead of minimizing the cost function from every initial condition $\bx$ as in (\ref{cost_1}), our proposed convex formulation relies on minimizing the following cost function averaged over all states $\bx\in \bX_0$. 
 \begin{eqnarray}
\min_{\bu} \;\;J(\mu_0)=\min_{\bu}\;\;\int_{\bX} V(\bx)d\mu_0(\bx) \label{eq:control_cost}
\end{eqnarray}
where $\mu_0$ is the measure capturing the distribution of the initial state. In particular, for the initial state of the vehicle in set $\bX_0$, we have measure $\mu_0$ supported on set $\bX_0$. The form of the cost function where $V(\bx)$ is averaged over the state $\bx\in \bX_0$ plays a fundamental role in the convex formulation of optimal navigation problem in the space of density.

In the rest of the paper, we will assume that $\mu_0$ is absolutely continuous with density function $h_0$, i.e.,$\frac{d \mu_0}{d \bx} = h_0(\bx)$. For example if $\mu_0$ is supported on initial set then $h_0(\bx)=\mathds{1}_{\bX_0}(\bx)$ i.e., indicator function of set $\bX_0$. The objective is to find the feedback controller $\bk(\bx)$ to minimize the cost function in (\ref{eq:control_cost}). Appropriate conditions on the initial measure $\mu_0$ are necessary to ensure the cost function is finite. We make sure that the density function $h_0$ is finite and positive semi-definite on $\bX$ and $h_0 \in {\cal L}_1(\bX)\cap{\cal C}^1(\bX)$.

Along with minimizing the cost function, it is also of interest to avoid certain obstacle sets, $\bX_u$, and limit the control authority, i.e., $|u_j|\leq L_j$. The obstacle avoidance constraints for almost every trajectory starting from the initial set $\bX_0$ are written as 
\[\int_{\bX} \mathds{1}_{\bX_u}(\bx(t))d\mu_0(\bx)=0,\;\;\forall t\geq 0 \]
where $\bx(t)$ is the solution of system (\ref{eq:sys_control}) starting from initial condition $\bx$. 
With the above definition, we can state the problem statement for optimal off-road navigation using the traversability map as given below.


\begin{problem}\label{problem_statement2} (Optimal off-road Navigation Problem)  Navigate almost every system trajectory for  (\ref{eq:sys_control}) starting from the initial set $\bX_0$ to the target set $\bX_T$ while avoiding the obstacle set $\bX_u$ such that cost for traversing is kept below some threshold ($\gamma$) and the following cost is minimized. 
\begin{subequations}\label{problem_B}
    \begin{alignat}{5}
    \min_{\bu}&\;\;\; \int_{\bX}\int_0^\infty \left(q(\bx(t))+\bu^\top \bR\bu \right)dt d\mu_0(\bx)\label{problem_a}\\
    {\rm s.t.}&\;\;\; \;\int_{\bX}\mathds{1}_{\bX_u}(\bx(t))d\mu_0(\bx)=0,\;\;\forall t\geq 0 \label{problem_b}\\
      & \int_{\bX} \int_0^\infty b(\bx(t)) dtd\mu_0(\bx)\leq  \gamma \label{problem_c} \\
      &|u_j|\leq L_j,\;\;\;j=1,\ldots m.\label{problem_d}\\
       & \dot \bx ={\bf f}(\bx)+\bg(\bx)\bu,\;\;\;\;\lim_{t\to \infty} \bx(t)\in \bX_T. \label{problem_e}
       \end{alignat}
\end{subequations}
where $q(\bx)$ is the distance from $\bx$ to the target set $\bX_T$. So the objective is to find the shortest path to the target set while keeping the traversability cost below a certain threshold $\gamma$. 
\end{problem}

The following section proves that the optimal off-road navigation problem, defined in Problem 1, can be written as a convex optimization over the density space. 
\section{Convex Formulation of Optimal off-Road Navigation}\label{main_section}
We make the following assumption on the control dynamical system and the controller for the system \eqref{eq:sys_control}. 
\begin{assumption}\label{assume1} For the optimal off-road navigation problem, we assume that the optimal control input is feedback in nature, i.e., $\bu=\bk(\bx)\in {\cal C}^1(\bX)$, such that the cost function corresponding to this input is finite. 
\end{assumption}
With the above assumption, we can write the feedback control system in the form
\begin{align}
    \dot \bx={\bf f}(\bx)+\bg(\bx)\bk(\bx)=: \bF_c(\bx)\label{closed_sys}.
\end{align}
Now, we will consider the following definition of almost everywhere (a.e.) stability as introduced in \cite{vaidya2008lyapunov}.
\begin{definition}\label{def_aeuniformstable}[Almost everywhere (a.e.) stability]  The equilibrium point or an attractor set of the system (\ref{closed_sys}) represented by $A$ is said to be a.e. stable w.r.t. measure $\mu_0\in {\cal M}(\bX)$ if 
\begin{eqnarray}
\mu_0\{\bx\in \bX: \lim_{t\to \infty} \bs_t(\bx)\notin A\}=0.\label{eq_aeunifrom}
\end{eqnarray}
\end{definition}
This paper's main results ensure that the feedback control obtained for the navigation is a.e. stable with respect to the target set of \eqref{closed_sys}. We next introduce the following definitions of linear operators \cite{lasota1998chaos}.
\begin{definition}[Koopman Operator]$\mathds{U}_t:{\cal L}_{\infty}(\bX)\rightarrow{\cal L}_{\infty}(\bX)$ for (\ref{closed_sys}) is given by
\begin{align}
    [\mathds{U}_t \phi](\bx) = \phi(\bs_t(\bx)), \label{eq:def_Koopman_operator}
\end{align}
where $\phi$ is a test function in the lifted function space $\mathcal{L}_{\infty}(\bX) \bigcap \mathcal{C}^1(\bX)$. The Koopman generator for \eqref{closed_sys} is given by
\begin{align}
    \lim_{t\to 0}\frac{[\mathbb{U}_t\phi](\bx)-\phi(\bx)}{t}={\bF_c}(\bx)\cdot \nabla \phi(\bx)=:{\cal U}_{\bF_c} \phi .\label{eq:def_Koopman_generator}
\end{align}
\end{definition}
\begin{definition}[Perron-Frobenius Operator] $\mathds{P}_t:{\cal L}_1(\bX)\to {\cal L}_1(\bX)$ for (\ref{closed_sys}) is given by
\begin{align}
    [\mathds{P}_t\varphi](\bx) = \varphi(\bs_{-t}(\bx)) \left| \frac{\partial \bs_{-t}(\bx)}{\partial \bx}\right|, \label{eq:def_PF_operator}
\end{align}
where $|.|$ stands for the determinant and $\varphi$ is a test function. The P-F generator for \eqref{closed_sys} is given by
\begin{align}
    \lim_{t\to 0}\frac{[\mathbb{P}_t\varphi](\bx)-\varphi(\bx)}{t}= -\nabla \cdot (\bF_c(\bx)\varphi(\bx))=:{\cal P}_{\bF_c} \varphi .\label{eq:def_PF_generator}
\end{align}
\end{definition}
These two operators are dual to each other as
\begin{align}
    \int_{\bX}[\mathbb{U}_t \phi](\bx)\varphi(\bx)d\bx=
\int_{\bX}[\mathbb{P}_t \varphi](\bx)\phi(\bx)d\bx. \label{eq:duality}
\end{align}
These two operators enjoy positivity and Markov properties used in the finite-dimension approximation.
\subsection{Obstacle avoidance constraints}
The first result of this paper allows us to write the obstacle avoidance constraints (\ref{problem_b}) in the integral form.   
\begin{lemma}\label{lemma_1}
For the dynamical system (\ref{closed_sys}), if
\begin{align}
\int_{0}^\infty \int_{\bX} \mathds{1}_{\bX_u}(\bs_t(\bx))d\mu_0(\bx)dt=0,\label{lemmaeq1}
\end{align} then 
\begin{align}
    \int_{\bX} \mathds{1}_{\bX_u}(\bs_t(\bx))h_0(\bx)d\bx=0,\;\;\;\forall t\geq 0,\label{contra}
\end{align}\label{lemma:time_unsafeset}
i.e., the amount of time system trajectories spends in the region $\bX_u$ starting from the positive measure set of initial condition corresponding to the initial set, $\bX_0$, with density $h_0(\bx)$ is equal to zero.
\begin{proof}
Proof by contradiction. 
Assume (\ref{contra}) is not true, i.e., there exists some time $t_0$ for which 
\[\int_{\bX} \mathds{1}_{\bX_u}(\bs_{t_0}(\bx))h_0(\bx)d\bx=\int_{\bX}[ \mathbb{U}_{t_0}\mathds{1}_{\bX_u}](\bx)h_0(\bx)d\bx>0\]
Then using the continuity property of the Koopman semi-group, we know there exists a $\Delta$ such that 
\[\int_{t_0}^{t_0+\Delta}\int_{\bX}[ \mathbb{U}_{t_0}\mathds{1}_{\bX_u}]((\bx))h_0(\bx)d\bx dt>0.\]
We have 
{\small
\begin{equation}
\begin{aligned}
    0<\int_{t_0}^{t_0+\Delta}\int_{\bX}[ \mathbb{U}_{t_0}\mathds{1}_{\bX_u}](\bx)h_0(\bx)d\bx dt\leq\\\int_{0}^{\infty}\int_{\bX}[ \mathbb{U}_{t_0}\mathds{1}_{\bX_u}](\bx)h_0(\bx)d\bx dt=0
\end{aligned}
\end{equation}}
\end{proof}
\end{lemma}
\begin{remark}
Following the results of Lemma \ref{lemma_1}, we can replace the obstacle avoidance constraints (\ref{problem_b}) as 
\begin{align}
\int_{0}^\infty \int_{\bX} \mathds{1}_{\bX_u}(\bs_t(\bx))d\mu_0(\bx)dt=0
\end{align} 
\end{remark}
\begin{remark}
In the rest of the paper, we will use the notation $\bX_1=\bX\setminus {\cal N}_\epsilon$, where ${\cal N}_\epsilon$ is the $\epsilon$ neighborhood of the target set. With no loss of generality, we assume the target set $\bX_T$ is locally stable with the domain of attraction containing ${\cal N}_\epsilon$. 
\end{remark}
\subsection{Convex formulation}
The following theorem presents the main result of this paper. 
\begin{theorem}\label{theorem3}
Under Assumption \ref{assume1}, the optimal off-road navigation problem can be written as the following convex optimization problem in terms of optimization variables 
$\rho \in{\cal L}_1(\bX_1)\cap {\cal C}^1(\bX_1,\mR_{\geq 0})$ and $\bar \brho\in{\cal C}^1(\bX_1,\mR)$
\begin{subequations}\label{convex_formulation}
\begin{align}
    J^\star=&\inf_{\rho,\bar \brho}\int_{\bX_1} q(\bx)\rho(\bx)+\frac{\bar{\brho}(\bx)^\top \bR\bar{\brho}(\bx)}{\rho(\bx)} \;d\bx \label{ocp_dens}\\
{\rm s.t}.\;\;\;
&\int_{\bX_1}\mathds{1}_{\bX_u}(\bx)\rho(\bx)d\bx=0 \label{feas223}\\
&\int_{\bX_1}b(\bx)\rho(\bx)d\bx\leq \gamma\label{feas222}\\
&|\bar{\rho}_j(\bx)| \le L_j \rho(\bx) \label{feas224}\\
&\nabla\cdot(\bff \rho+\bg \bar \brho)=h_0,\;\;{\rm a.e.}\; \bx\in \bX_1\label{feas221}
\end{align}
\end{subequations}
The solution to the above optimization problem is used to recover the optimal feedback control input as follows:
\begin{align}
    \bk^\star(\bx) = \frac{\bar{\brho}^\star(\bx)}{\rho^\star(\bx)} \label{optimal_control}
\end{align}
where $(\rho^\star,\bar{\brho}^\star)$ are the solution of \eqref{convex_formulation}.
\end{theorem}
The proof is provided in the Appendix.

\begin{remark}
The convex structure of the optimization problem presented in \eqref{convex_formulation} comes from the fact that the decision variables $\rho$ and $\bar{\brho}$ enter the cost quadratically and the constraints linearly and hence led to the infinite-dimensional convex optimization problem. We discuss the data-driven approach for the finite-dimensional approximation of the infinite-dimensional problem in Section~\ref{comp_framework}.  
\end{remark}
\subsection{Control constraints}
Other constraints on the state and control input can be written convexly in terms of the optimization variables $\rho$ and $\bar \brho$. One such constraint is the curvature constraint. For example, consider a Dubin's car model
\[\dot x_1=u_1 \cos \theta,\;\;\dot x_2=u_1 \sin \theta ,\;\;\dot \theta=u_2.\]
Aside from the kinematic constraints imposed by nonholonomy, most often, the additional constraint on the radius of curvature of the paths of the vehicle must be considered \cite{balluchi1996path}. The curvature constraints themselves will be a function of the terrain properties. The curvature constraints close the kinematic model of Dubin's car to reality. We can capture the curvature constraints as follows.
\begin{align}
    \frac{|u_2|}{|u_1|} \le \frac{1}{\mathds{C}} \label{curvature_const}
\end{align}
where $u_2$ is the angular velocity of Dubin's car model. The constraint in \eqref{curvature_const} ensures that the control design is realistic and that the angular velocity of the kinematic model is bounded. The \eqref{curvature_const} can be easily added to the optimization problem of theorem \ref{theorem3} convexly as follows. Following the general formulation specialized for two input systems, we can write $u_1=\frac{\bar \rho_1}{\rho}$ and $u_2=\frac{\bar \rho_2}{\rho}$, hence
\begin{align}
    \frac{|u_2|}{|u_1|} \le \frac{1}{\mathds{C}} \implies \mathds{C} |\bar \rho_2|-|\bar \rho_1|\leq 0
\end{align}
which is linear and hence convex in $\bar\rho_1$ and $\bar \rho_2$.
\section{Computational Framework}\label{comp_framework}
In this section, we utilize Naturally structured dynamic mode decomposition (NSDMD) \cite{huang2018data} and provide a modification approximation, namely Approximate NSDMD. Then we formulate our problem with the given computational framework.
\subsection{Data-Driven Approximation: Approximate NSDMD}\label{2b}
The NSDMD algorithm introduced in \cite{huang2018data} incorporates the natural properties of the linear operators, namely the positivity and the Markov properties. However, using the NSDMD algorithm is computationally expensive. Also, executing the NSDMD algorithm to find the P-F operator involves convex optimization with linear constraints and does not admit an analytical solution. Our proposed transfer operator theoretical framework for the control design relies on these operators' positivity and Markov properties. In our proposed modification of the NSDMD algorithm, we gain numerical efficiency at the expense of preserving these properties approximately. We call this modified algorithm ``Approximate NSDMD (A-NSDMD),'' and it can be solved analytically as it is a least square problem, as shown below. For the continuous-time dynamical system (\ref{eq:sys_control}), consider snapshots of the data set obtained as time-series data from single or multiple trajectories.
\begin{eqnarray}
{\mX}= [\bx_1,\bx_2,\ldots,\bx_Q],\;\;\;\;{\mY} = [\mathbf{y}_1,\mathbf{y}_2,\ldots,\mathbf{y}_Q] ,\label{data}
\end{eqnarray}
where $\bx_i\in \bX$ and $\mathbf{y}_i\in \bX$ such that $\mathbf{y}_i=\bs_{\Delta t}(\bx_i)$. In our proposed data-driven computation, we obtain the pair of two consecutive snapshots using the system with no control input $(\bs^0_{\Delta t}(\bx_i))$ and with the step input $(\bs^1_{\Delta t}(\bx_i))$. 
Let ${\bPsi}=[\psi_1,\ldots,\psi_P]^\top$ be the set of basis functions.
We obtain the finite-dimensional approximation of the Koopman operator ($\bf U$) through EDMD as the result of the following least square problem. 
\begin{equation}\label{edmd_op}
\min\limits_{\bf U}\parallel {\bf G}{\bf U}-{\bf A}\parallel_F,
\end{equation}
\begin{eqnarray}\label{edmd1}
{\bf G}=\frac{1}{Q} \bPsi({\mX}) \bPsi({\mX})^\top,\;\;\;
{\bf A}=\frac{1}{Q} \bPsi({\mX}) \bPsi({\mY})^\top,
\end{eqnarray}
with ${\bf U},{\bf G},{\bf A}\in\mathbb{R}^{P\times P}$, $\|\cdot\|_F$ stands for Frobenius norm. The analytical solution of the above least square problem is  
\begin{eqnarray}
\bU=\bf{G}^\dagger \bf{A}\label{edmd_formula}.
\end{eqnarray}
The EDMD algorithms provide convergence results with respect to the number of data points and basis functions \cite{korda2018convergence,klus2020eigendecompositions}.
In this paper, we work with Gaussian Radial Basis Functions (RBF), which are the positive basis. The positivity constraints utilized in NSDMD are avoided in the A-NSDMD based on the approximated positivity of the Koopman operator when RBFs are used. This means that all the negative elements of the Koopman matrix obtained are 1e-6 or lower. Under the assumption that the basis functions are positive, the A-NSDMD obtains $\hat \bU$ from $\bU$ using row normalization, i.e., the entries of the matrix $\hat \bU$ are obtained as:
\begin{align}
    [\hat {\bU}]_{ij}=\frac{[\bU]_{ij}}{\sum_j[\bU]_{ij}} \label{pos}
\end{align}
where $[\hat\bU]_{ij}$ is the $(i,j)^{th}$ entry of the matrix $\hat \bU$. The above modification helps us to avoid Markov constraints. As a result, the A-NSDMD is just a least square problem with an analytical solution as given by \eqref{edmd_formula}. We obtain the P-F matrix as  
${\hat{\bf P}} = \hat{\bf U}^{\top}$. The generator of the P-F operator is obtained as
\begin{eqnarray}
{\cal P}_{\bff}\approx \frac{\hat{\bf P}-{\bf I}}{\Delta t}=:\bM.\label{PF_approximation}
\end{eqnarray}
where $\textbf{I}$ is the identity matrix. Note that we did not strictly enforce the positivity property on the linear operators in the above construction. However, the numerical evidence suggests that the entries of the matrix $\bU$ obtained using the EDMD algorithm with a positive basis function are predominantly positive. Table~\ref{table:comp time} compares the computation time of calculating the P-F operator for different numbers of RBFs utilized in lifting the dynamics. This provides empirical proof of the computational efficiency of A-NSDMD over NSDMD.
\subsection{ Approximation of  Optimization Problem}
This section discusses finite-dimensional approximation of the infinite-dimensional optimal navigation problem (\ref{ocp_dens})-(\ref{feas222}). We use the approximation of the generator for the vector field $\bff$ and $\bg$. The generator approximation of the two vector fields is given as follows
\begin{align}
    \mathcal{P}_{\bff} \approx \textbf{M}_0,\;\;\;\; \mathcal{P}_{\bg} \approx \textbf{M}_1. \label{gen_approx}
\end{align}
The generator approximation of $\mathcal{P}_{\bff}$ starts with obtaining open-loop time series data $\{\bx_k^0\} = \{\bx_0^0,\bx_1^0,\dots,\bx_Q^0\}$ by substituting $\bu={\bf0}$ in \eqref{eq:sys_control}. Similarly, the approximation of $\mathcal{P}_{\bg_i}$ is obtained by getting $\{\bx_k^i\} =\{\bx_0^i,\bx_1^i,\dots,\bx_Q^i\}$  from substituting $\bu=\bu_i$ in \eqref{eq:sys_control} where $\bu_i$ is a column vector consisting of all zeros except at $i^{th}$ position where the element value is 1. Then we calculate the matrices $\hat{\bU}^0$ and $\hat{\bU}^i$ by utilizing \eqref{edmd_op}-\eqref{pos}. The P-F matrix is then obtained as $\hat{\bP}^0= \hat{\bU}^{0\top}$ and  $\hat{\bP}^i= \hat{\bU}^{i\top}$. Then we calculate $\textbf{M}_0$ and $\textbf{M}_i$ as follows:
\begin{align}
{\cal P}_{\bff}\approx \frac{\hat{\bf P}^0-{\bf I}}{\Delta t}=:\bM_0,\;\; {\cal P}_{\textbf{f+g$_i$}} \approx \frac{\hat{\bf P}^i-{\bf I}}{\Delta t}
\end{align}
\begin{align}
{\rm and}\;\;{\cal P}_{\textbf{g$_i$}} = {\cal P}_{\textbf{f+g$_i$}} - {\cal P}_{\bff} =: \bM_i
\end{align}
Approximate NSDMD algorithm outlined in Section \ref{2b} is used to approximate with $\bPsi(\bx) = [\psi_1(\bx),\psi_2(\bx),\dots,\psi_P(\bx)]^\top$ as the basis functions. Let $\bar{\brho}(\bx) = [\bar{\rho}_1(\bx),\dots,\bar{\rho}_m(\bx)]^\top$ and ${\bf w} = [{\bf w}_1,\dots,{\bf w}_m]^\top $. Now let us express the following terms as combinations of basis functions
\begin{align}
    h_0 = \bPsi^\top {\bf m},\; \rho = \bPsi^\top {\bf v},
    \;{\bar{\rho}}_j = \bPsi^\top {\bf w}_j. \label{term_approx}
\end{align}
Assuming $\bR$ is the identity matrix, we perform the following finite-dimensional approximation
\begin{subequations}
\begin{align}
    \inf_{\rho,\bar \brho}&\int_{\bX_1}\left[q(\bx)\rho(\bx)+\frac{\bar{\brho}(\bx)^\top\bar{\brho}(\bx)}{\rho(\bx)} \;\right]d\bx \approx \nonumber\\
    &\underset{{\bf w}_j,{\bf v}}{\min}\;{\bf D_1} {\bf v}+{\bf w}^\top{\bf D_2} \frac{{\bf w}}{{\bf v}}
\end{align}
\end{subequations}
where ${\bf D}_1 = \int_{\bX_1}q(\bx)\bPsi^\top(\bx)d\bx$ , ${\bf D}_2 = \int_{\bX_1} \bPsi(\bx)\bPsi^\top(\bx) \;d\bx$ and division is assumed element-wise.

Similarly, we can write
\begin{align}
    \int_{\bX_1}\mathds{1}_{\bX_u}(\bx)\rho(\bx)d\bx\approx  \left[\int_{\bX_1} \mathds{1}_{\bX_u}(\bx)\bPsi^\top(\bx)d\bx \right]\bf v = {\mathbf{d_1}} \bf v\\
        \int_{\bX_1}b(\bx)\rho(\bx)d\bx\approx  \left[\int_{\bX_1}b(\bx)\bPsi^\top(\bx)d\bx \right]\bf v = {\mathbf{d_2}} \bf v
\end{align}
where ${\mathbf{d_1}}:=\int_{\bX_1}\mathds{1}_{\bX_u}(\bx)\bPsi^\top d\bx$ and ${\mathbf{d_2}}:=\int_{\bX_1}b(\bx)\bPsi^\top d\bx$.
\begin{assumption}\label{assume_gaussian}
We consider all the basis functions to be positive and linearly independent. 
\end{assumption}
\begin{remark}
\label{remark_positivebasis}
The Gaussian Radial Basis Functions (RBF) is used in this paper for the simulation results i.e., $\psi_k(\bx)=\exp(-\frac{\parallel \bx-{\bf c}_k\parallel^2}{2\sigma^2})$
where $\bc_k$ is the center of the $k^{th}$ Gaussian RBF. 
\end{remark}
Therefore the ONP in Theorem \ref{theorem3} can be written as 
\begin{eqnarray}
&\underset{{\bf w}_j,{\bf v}}{\min}\;{\bf D_1} {\bf v}+{\bf w}^\top{\bf D_2} \frac{{\bf w}}{{\bf v}} \nonumber\\
\text{s.t.}&-\bM_0{\bf v}-\sum_j \bM_j{\bf w}_j=\bf m\nonumber\\
&\mathbf{d_1} {\bf v} = 0,\;\;\;
\mathbf{d_2} {\bf v}\leq \gamma,\;\;\;
{|{\bf w}_j|} \le L_j \bf v \nonumber
\end{eqnarray}
\section{Simulation Results}
The centers of the RBF used in the simulation results are chosen to be uniformly distributed in the state space at a distance of $d$. The RBF's standard deviation ($\sigma$) is selected as $0.5\;d$. We perform all the simulations using MATLAB R2021b on a Dell computer with 64 GB Random Access Memory (RAM) and an Intel(R) i7-10700K processor (3.80 GHz). In this paper, we use the data to do the linear approximation of the single integrator dynamics and Dubin's car model. The CVX toolbox is used for solving an optimization problem. The computational time to obtain the proposed algorithm's feedback control ranges from 5-15 minutes.\\

\noindent \textbf{Example 1.} Consider the single integrator dynamics 
\begin{align}
    \dot{x}_1 = u_1, \;\;\; \dot{x}_2 = u_2  \label{eq:single_integrator}
\end{align}

The initial and terminal sets are labeled as $\bX_{01},\bX_{02}$ and $\bX_{T1}$. These sets are described as follows. 
\begin{itemize}
    \item $\bX \triangleq \{\bx \in \mR^2 : -3 \le x_1 \le 3,\; -3 \le x_2 \le 3 \}$
    \item $\bX_{01} \triangleq \{\bx \in \mR^2 : 1 - (x_1 - 1.5)^{2} - (x_2 + 2.25)^2 \ge 0 \}$
    \item $\bX_{02} \triangleq \{\bx \in \mR^2 : 1 - (x_1 + 1.5)^{2} - (x_2 + 2.25)^2 \ge 0 \}$
    \item $\bX_{T1} \triangleq \{\bx \in \mR^2 : 1 - x_1^{2} - (x_2 - 2)^2 \ge 0 \}$
    \item $\bX_{u1} \triangleq \{\bx \in \mR^2 : 1 - x_1^{2} - x_2^2 \ge 0 \}$
\end{itemize}
The linear operators are approximated using $625$ RBF and the time-step of discretization $\Delta t = 0.01$. The value of $\gamma$ is 0.6. We have obtained the results for the traversability maps labeled by $b_1(\bx)$. The optimal navigation problem given in \eqref{ocp_dens} - \eqref{feas221} can be utilized to avoid any particular predetermined region, such as an enemy camp, lake, or pits. In Fig.~\ref{si_b1_obs}, we show plot for initial conditions in initial sets $\bX_{01}$ and $\bX_{02}$ navigating successfully to terminal set $\bX_{T1}$ with traversability map given by $b_1(\bx)$. We notice that the trajectories navigate through a region where the traversability map takes smaller values while trying to avoid hard obstacle set $\bX_{u1}$.\\

\begin{figure}[ht]
  \centering
  \includegraphics[width=0.9\linewidth]{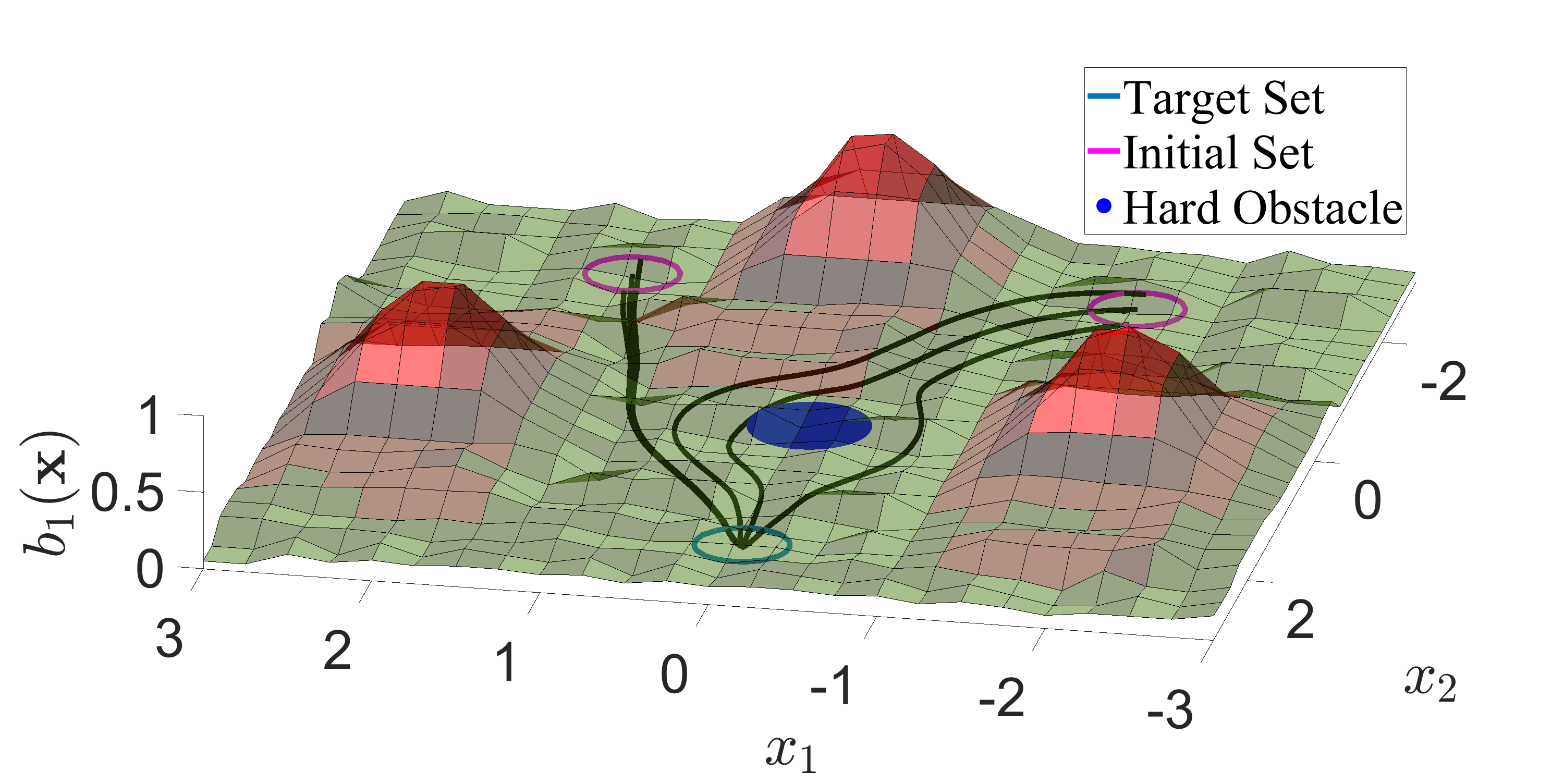}
\caption{Optimal navigation with initial points centered around $(1.5,-2.25)$ in initial set $\bX_{01}$ and centered around $(-1.5,-2.25)$ in initial set $\bX_{02}$ to terminal set $\bX_{T1}$ with traversability map given by $b_1(\bx)$.The traversability map includes hard obstacle centered around $(0,0)$.}
\label{si_b1_obs}
\end{figure}
\begin{figure}[ht]
  \centering
  \includegraphics[width=0.9\linewidth]{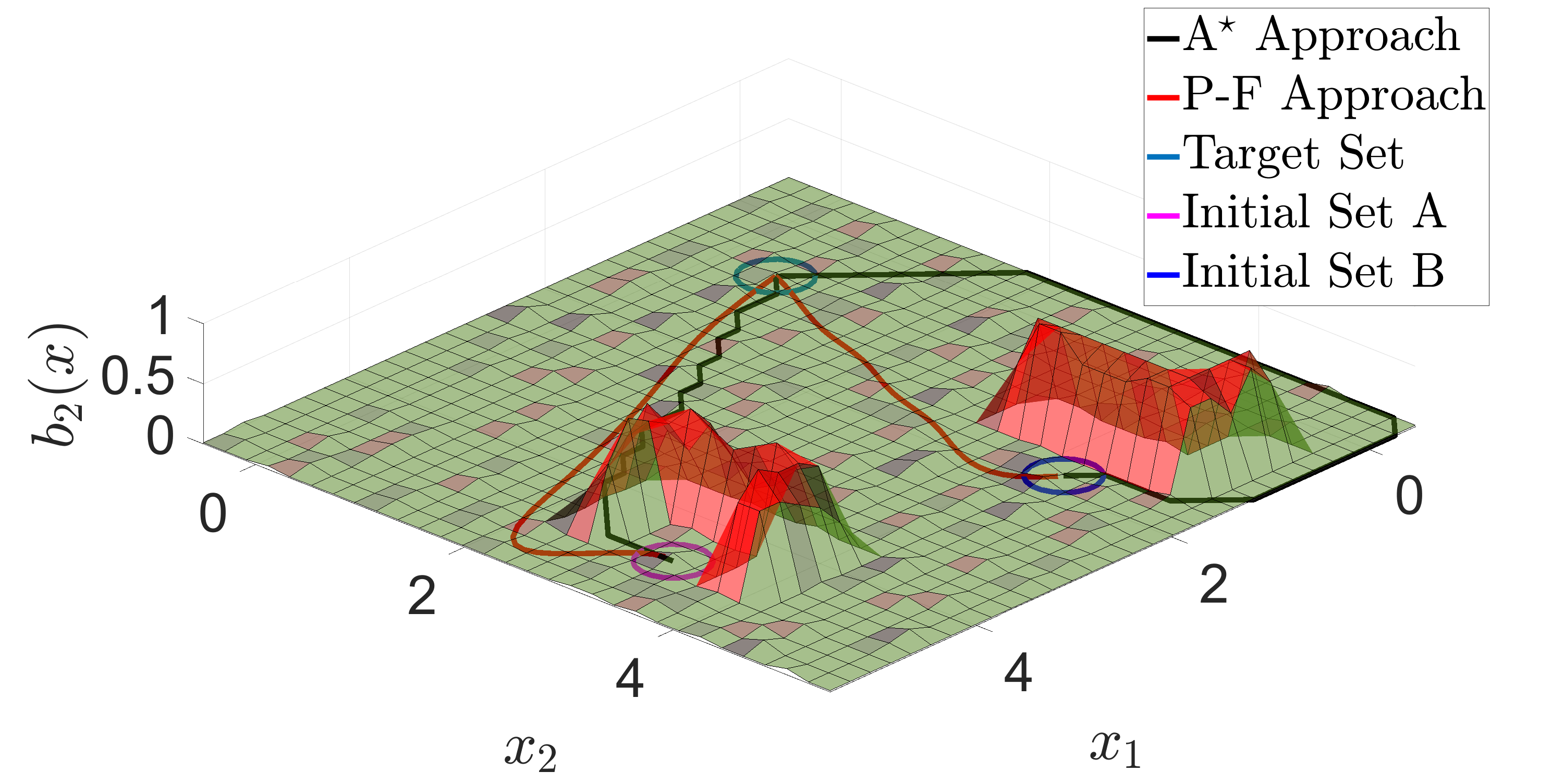}
\caption{Comparison of P-F approach with A$^{\star}$ algorithm with traversability map given by $b_2(\bx)$. The trajectories start from two initial points: point A $(4.5,3.2)$ and point B $(1.7,4.2)$ and end at target point given by $(0.7,0.5)$.}
\label{astar}
\end{figure}

\noindent \textbf{Example 2.} Consider the Dubin's car dynamics 
\begin{align}
    \dot{x}_1 = u_1\;\cos\theta, \;\;\; \dot{x}_2 = u_1 \; \sin\theta \;\;\;  \dot\theta= u_2 \label{eq:dubin}
\end{align}
The initial and terminal sets are labeled as $\bX_{03}$ and $\bX_{T3}$. These sets are described as follows.  
\begin{itemize}
    \item $\bX \triangleq \{\bx \in \mR^3 : -3 \le x_1 \le 9,\; -3 \le x_2 \le 9,\; -3 \le x_3 \le 3 \}$
    \item $\bX_{03} \triangleq \{\bx \in \mR^2 : 1 - (x_1 - 5.5)^{2} - (x_2 - 7.5)^2 \ge 0 \}$
    \item $\bX_{T3} \triangleq \{\bx \in \mR^2 : 1 - x_1^{2} - x_2^2 \ge 0 \}$
    \item $\bX_{u5} \triangleq \{\bx \in \mR^2 : 1 - (x_1 - 3)^{2} - (x_2 - 4)^2 \ge 0 \}$
\end{itemize}
The linear operators are approximated using $2500$ RBF and the time-step of discretization $\Delta t = 0.01$. The value of $\gamma$ is $0.7$. We have obtained the results for the traversability map given by $b_3(\bx)$.  We show the optimal navigation plots with a hard obstacle set $\bX_{u5}$  in Fig. \ref{obstacle_1}. The control plots for the trajectories given in Fig. \ref{obstacle_1} are displayed in Fig. \ref{control1}. The control limits are set to be $\pm 3$ and are seen to be met by the control.\\

The problem of data-driven approximation of linear operators is an active area of research. There are mainly two sources of errors: first, due to the finite amount of data, and second, due to the finitely many basis functions used in the approximation. In \cite{klus2020eigendecompositions,korda2018convergence,hou2021sparse}, sample complexity results for the finite-dimensional approximation of the linear operator are developed, where for a fixed number of basis functions, the approximation error is shown to decay as $\frac{1}{\sqrt{M}}$, where $M$ is the number of data points. Characterizing approximation error as the function of the number of basis functions is difficult as the results will depend on the underlying system dynamics. The sample complexity-based error bounds for the approximation of linear operator can be used to characterize the convergence rate of the solution obtained from the finite-dimensional approximate optimization problem to the true optimal solution \cite{vaidya2022data} due to the proposed convex formulation of the optimization problem.  \\

The computational framework for the finite-dimensional approximation of linear operators suffers from scalability issues for a system involving high dimensional state space. However, there are ways to address the computational burden for the system with reasonable (5-10) state-space dimensions. In particular, methods based on exploiting the sparsity structure of the linear operators, spectral properties of linear operators, and using the back-stepping approach for control design are currently being investigated to extend the applicability of the results to the system of practical interest \cite{sinha2019computation,hongzesparse,vaidya2022spectral}.


\begin{table}[ht]
\centering
\caption{P-F operator Computational time readings}
\begin{tabular} {ccc}
\hline
\# RBF & NSDMD &Approximate NSDMD\\
\hline
\hline
500 & 18.7 sec & 10.2 sec\\
1000 & 99.88 sec & 53.86 sec\\
1500 & 400.2 sec & 162.97 sec\\
\hline
\hline
\end{tabular}
\label{table:comp time}
\end{table}

\begin{figure}[ht]
  \centering
  \includegraphics[width=0.9\linewidth]{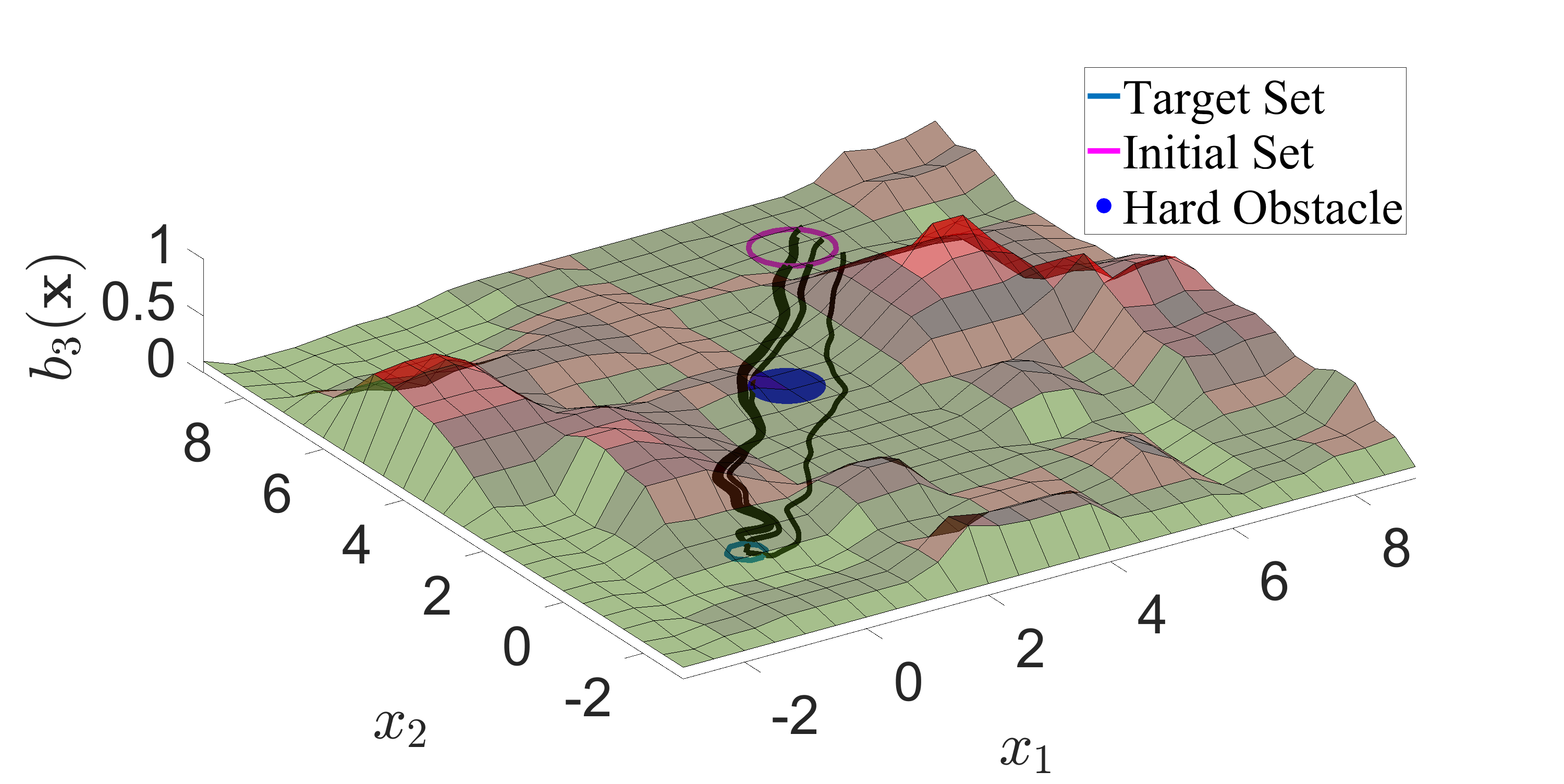}
\caption{Optimal navigation with initial points centered around $(5.5,7.5)$ in initial set $\bX_{04}$ to terminal set $\bX_{T3}$ with traversability map given by $b_3(\bx)$.The traversability map includes hard obstacles centered around $(3,4)$.}
\label{obstacle_1}
\end{figure}

\begin{figure}[ht]
  \centering
  \includegraphics[width=0.9\linewidth]{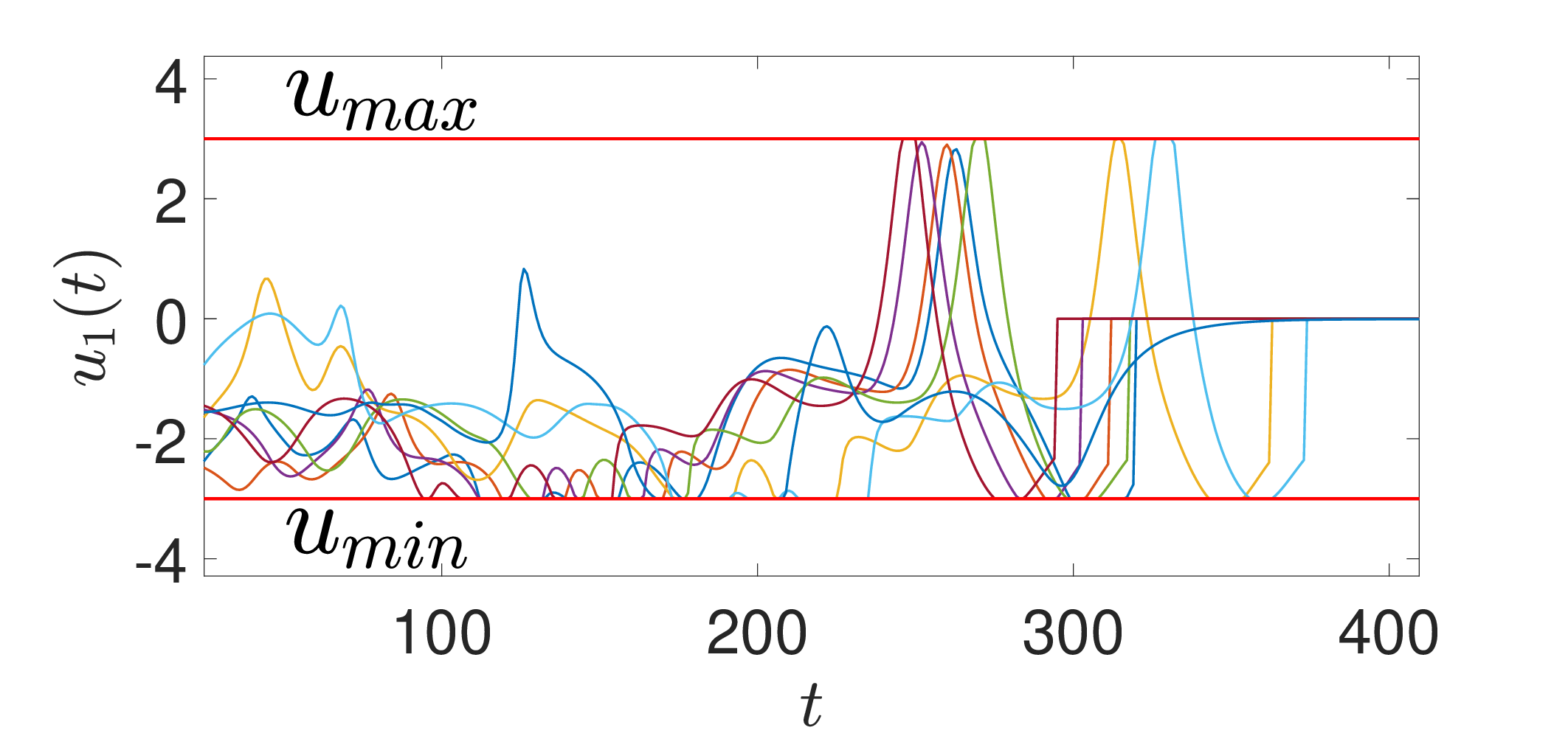}
\caption{Optimal control plots  for $u_1$ corresponding to different initial conditions with the configuration of initial and final sets as in  Fig. \ref{obstacle_1}.}
\label{control1}
\end{figure}

We also compare our approach with A$^{\star}$ algorithm to show its efficiency over the existing approach. We use MATLAB 2021b as the simulation platform for the comparison. The A$^{\star}$ algorithm \cite{hong2021improved} is a search-based algorithm for path planning which uses heuristic functions for selective search in the workspace. The cost function of A$^{\star}$ algorithm consists of two terms: $g(\bx)$, the actual cost, and $h(\bx)$, the heuristic cost. Here, $g(\bx)$ represents the transition cost from the initial state to the current state, whereas $h(\bx)$ represents the cost to go from the current state to the target state. When A$^{\star}$ is implemented with a grid map, then $g(\bx)$ would represent the grid values. For our comparison purposes, we selected $g(\bx)$ to be equal to the traversability map $b(\bx)$ and $h(\bx)$ to be equal to the Euclidean distance between the current state and target state. The trajectory comparison between  A$^{\star}$ algorithm and the P-F approach is shown in Fig.~\eqref{astar}. Here,$\gamma = 10$ for the P-F-based approach. Table~\ref{table:comparison1} contains the traversability cost of trajectories obtained from A$^{\star}$ algorithm and the P-F approach. The traversability cost is calculated by summing the traversability values, $b(\bx(t))$, along the trajectories generated by the A$^\star$ and P-F based approach. 

 \begin{table}[ht]
        \centering
        \caption{Comparison between A$^{\star}$ and P-F approach}
        \begin{tabular} {ccccc}
        \hline
        Traversability & Starting & Target & \multicolumn{2}{c}{Traversability Cost}\\
        \hhline{~~~--}
        map & point & point & A$^{\star}$ & P-F approach\\
        \hline
        \hline
        \multirow{1}{*}{$b_2(\bx)$ Fig.~\ref{astar}} & (4.5,3.2) set A&(0.7,0.5)& 8.24 & 6.1\\
        \hhline{~----}
        & (1.7,4.2) set B&(0.7,0.5)&9.34&4.3\\
        \hline
        \hline
        \end{tabular}
        \label{table:comparison1}
    \end{table}
We observe that our propsoed approach performs better than the A$^{\star}$ algorithm as it leads to lower cost of traversability. 

\section{Conclusion}
The problem of navigation on off-road terrain is considered. We use traversability measures to describe the relative degree of difficulty of navigation. 
A convex formulation for the optimal navigation problem is constructed using the traversability information of the terrain. The convex formulation leads to an infinite-dimensional convex optimization problem for navigation. Furthermore, we utilize the data-driven approximation of the linear P-F operator for the finite-dimensional approximation of the optimization problem. Finally, simulation results are showcased to show the validity of the proposed method. Future research efforts will incorporate the framework's uncertainty arising from vehicle dynamics and vehicle-terrain interaction.
\section{Appendix}
\begin{proof}
We use the fact that $d\mu_0=h_0(\bx)d\bx$, and the  definition of the Koopman operator to write the cost function (\ref{problem_a}) as 
\begin{align}
J = \int_{\bX_1} \int^{\infty}_0 [\mathds{U}_t(q + \bk^\top \bR \bk)](\bx) \;dth_0(\bx) d\bx \nonumber\\
=\int_{\bX_1} \int^{\infty}_0 (q + \bk^\top \bR \bk)(\bs_{t}(\bx)) \;dth_0(\bx) d\bx \nonumber
\end{align}
Performing change of variable $\by=\bs_t(\bx)$ or $\bx=\bs_{-t}(\by)$ and using the definition of P-F operator, we can write the above as 
\begin{align}
J &= \int_{\bX_1} \int^{\infty}_0 (q + \bk^\top \bR \bk)(\bx)[\mathds{P}^c_th_0] \;dt d\bx\nonumber\\
&=\int_{\bX_1} (q + \bk^\top \bR \bk)\rho(\bx)d\bx \label{cost11}
\end{align}
where 
\begin{align}\rho(\bx):=\int_0^\infty [\mP_t h_0](\bx)dt.\label{rho_def}\end{align} Since the cost function is assumed to be finite (Assumption \ref{assume1}), the $\rho(\bx)$ is well-defined for a.e. $\bx$ and is an integrable function. The function $[\mP_t h_0](\bx)$ is a uniformly continuous function of time which can be inferred from the definition of the P-F operator and the assumption made on function $h_0$. Hence using Barbalat Lemma \cite[pg. 269]{barbalat1959systemes} we have 
\begin{eqnarray}
\lim_{t\to \infty}[\mP_t h_0](\bx)=0,\label{convergence1}
\end{eqnarray}
for a.e. $\bx$.
Furthermore, since the P-F operator is a positive operator 
and $h_0(\bx)\geq 0$, we have $\rho(\bx)\geq 0$. Hence (\ref{cost11}) can be written as 
\begin{align}
J(\mu_0) &=\int_{\bX_1} (q(\bx)\rho(\bx)+ \frac{\bar \brho(\bx)^\top\bR \bar \brho(\bx)}{\rho(\bx)})d\bx \label{cost111}
\end{align}
where $\bar \brho(\bx):=\rho(\bx)\bk(\bx)$. This shows that we can write the cost function in (\ref{problem_B}) in the form (\ref{convex_formulation}). We next discuss the constraints. Following the results of Lemma \ref{lemma_1}, we know that (\ref{lemmaeq1}) implies (\ref{contra}). Hence, the (\ref{problem_b}) are implied by 
\begin{align}
\int_{0}^\infty \int_{\bX} \mathds{1}_{\bX_u}(\bs_t(\bx))d\mu_0(\bx)dt=0\label{ll}
\end{align}
Again performing change of variables $\by=\bs_t(\bx)$ and using the definition of P-F operator, (\ref{ll}) can be written as 
\[\int_{\bX_1}\mathds{1}_{\bX_u}(\bx)\rho(\bx)d\bx=0\]
with $\rho(\bx)$ as defined in (\ref{rho_def}). Similarly, it follows that (\ref{problem_c}) can be written as
\[\int_{\bX_1}b(\bx)\rho(\bx)d\bx\leq\gamma. \]
The control constraints $|u_j|\leq L_j$ for $j=1,\ldots, m$ can be written as $|\bar\rho_j(\bx)|\leq L_j \rho(\bx)$ follows from the fact that $\bu=\bk(\bx)=\bar \brho(\bx)/\rho(\bx)$ and $\rho(\bx)$ is positive. We next show that the $\rho(\bx)$ as defined in (\ref{rho_def}) satisfies the constraints (\ref{feas221}). 
Substituting (\ref{rho_def}) in the constraint of (\ref{feas221}), we obtain
\begin{eqnarray}
&&\nabla\cdot ({\bf F}_c(\bx) \rho(\bx))=\int_0^\infty \nabla\cdot ({\bf F}_c(\bx)[\mathbb{P}_t h_0](\bx)) dt\nonumber\\
&=&\int_0^\infty -\frac{d}{dt}[\mathbb{P}_t h_0](\bx)dt=-[\mathbb{P}_th_0](\bx)\Big|^{\infty}_{t=0}=h_0(\bx),\nonumber\\\label{eq11}
\end{eqnarray}
In deriving (\ref{eq11}) we have used the infinitesimal generator property of P-F operator Eq. (\ref{eq:def_PF_generator}) and the fact that $\lim_{t\to \infty} [\mathbb{P}_t h_0](\bx)=0$  following (\ref{convergence1}). Next, we show that the target set $\bX_T$ is a.e. stable w.r.t measure $\mu_0$ supported on the initial set $\bX_0$. Consider the set $S_\ell$ 
\[S_\ell=\{\bx\in \bX_1: \bs_t(\bx)\in \bX_1,\;{\rm for\;some}\;t>\ell\}\]
and let $S=\bigcap_{\ell=1}^\infty S_\ell$.
The set $S$ contains points, some of whose limit points lie in $\bX_1$, and for almost every stability of the target set, we need to show that $\mu_0(S)=0$. From the construction of the set $S$, it follows that $\bs_t(S)=S$, where $\bs_t(S)=\{\bs_t(\bx) :\bx\in S \}$.
\begin{align}\mu_0(S)=\int_{\bX_1}\mathds{1}_{S}(\bx)h_0(\bx)d\bx=\int_{\bX_1}\mathds{1}_{S}(\bs_t(\bx))h_0(\bx)d\bx\nonumber\\
=\int_{\bX_1}\mathds{1}_{S}(\bx)[\mP_t h_0](\bx)d\bx
\end{align}
where we have use the fact that $\bx\in S$ iff $\bs_t(\bx)\in S$. Since the above is true for all $t\geq 0$, we obtain using dominated convergence theorem
\[\mu_0(S)=\lim_{t\to \infty}\int_{S}[\mP_t h_0](\bx)d\bx=\int_{S}\lim_{t\to \infty}[\mP_t h_0](\bx)d\bx=0\]
\end{proof}
\bibliographystyle{IEEETRAN}
\bibliography{reference}
\end{document}